\documentclass[10pt,reqno]{amsart}
\usepackage{amscd}
\usepackage{amsthm}
\usepackage{amstext}
\usepackage{amssymb}
\usepackage{amsmath,amscd}
\usepackage{enumerate}
\usepackage{mathrsfs}
\usepackage[shortalphabetic]{amsrefs}
\usepackage{multirow}
\usepackage{tikz}

\numberwithin{equation}{section}

\newcommand{\cas}[1]{
  \begin{cases}
    #1
  \end{cases}
}
\newcommand{\BH}{\textnormal{BH}}
\newcommand{\DR}{\textnormal{DR}}
\newcommand{\End}{\textnormal{End}}

\newcommand{\Fr}{\textnormal{Fr}}
\newcommand{\GL}{\textnormal{GL}}
\newcommand{\Id}{\textnormal{Id}}
\newcommand{\TFr}{\textnormal{TFr}}
\newcommand{\ext}{\textnormal{ext}}
\newcommand{\orb}{\textnormal{orb}}
\newcommand{\ord}{\textnormal{ord}}

\newtheorem{theorem}{Theorem}[section]
\newtheorem{proposition}[theorem]{Proposition}
\newtheorem{lemma}[theorem]{Lemma}
\newtheorem{corollary}[theorem]{Corollary}

\theoremstyle{definition}

\newtheorem{remark}[theorem]{Remark}
\newtheorem{example}[theorem]{Example}

\newcommand{\cB}{\mathcal B}
\newcommand{\cC}{\mathcal C}

\newcommand{\cR}{\mathcal R}
\newcommand{\cS}{\mathcal S}
\newcommand{\cW}{\mathcal W}

\newcommand{\C}{\mathbb C}
\newcommand{\F}{\mathbb F}
\newcommand{\K}{\mathbb K}
\newcommand{\N}{\mathbb N}

\newcommand{\Z}{\mathbb Z}

\title{$p$-adic Berglund-H\"ubsch Duality}

\author{Marco Aldi}
\address{Department of Mathematics and Applied Mathematics\\
Virginia Commonwealth University\\
Richmond, VA 23284, USA}

\author{Andrija Peruni\v{c}i\'{c}}
\address{Department of Mathematics and Statistics \\
Queen's University\\
Kingston, ON K7L 3N6, Canada}

\begin{document}

\begin{abstract}
  Berglund-H\"ubsch duality is an example of mirror symmetry between
  orbifold Landau-Ginzburg models. In this paper we study a
  D-module-theoretic variant of Borisov's proof of Berglund-H\"ubsch
  duality. In the $p$-adic case, the D-module approach makes it
  possible to endow the orbifold chiral rings with the action of a
  non-trivial Frobenius endomorphism. Our main result is that the
  Frobenius endomorphism commutes with Berglund-H\"ubsch duality up to
  an explicit diagonal operator.
\end{abstract}

\maketitle

\section{Introduction}
\noindent 
Berglund-H\"ubsch duality was originally introduced \cite{BH} as a
generalization of the Greene-Plesser construction \cite{GP} of mirror
pairs. Let $W(x)\in \C[x]=\mathbb C[x_1,\ldots ,x_n]$ be an invertible
polynomial defining a Calabi-Yau hypersurface $X$ and let
$G\subset (\mathbb C^*)^n$ be a group fixing $W$. Then the
Berglund-H\"ubsch dual of the orbifold of $X$ by $G$ is the
hypersurface $X^T$, defined by the ``transpose'' invertible polynomial
$W^T(x)\in \mathbb C[x]$, orbifolded by an explicitly constructed
group $G^T\subset (\mathbb C^*)^n$ fixing $W^T$. As shown in
\cite{Kre} and \cite{Kra}, the construction of Berglund and H\"ubsch
can be further generalized to Landau-Ginzburg models with invertible
potentials (not necessarily of Calabi-Yau type) as follows. For any
invertible polynomial $W(x)$ the bigraded chiral ring of the orbifold
Landau-Ginzburg model $(W(x),G)$ is isomorphic to the (twisted) chiral
ring of the orbifold LG model $(W^T(x),G^T)$.

In the context of the vertex-algebra approach to mirror symmetry
\cite{B1}, Borisov \cite{B2} has shown that, as an isomorphism of
bigraded vector spaces (that is, disregarding the multiplicative
structure), Berglund-H\"ubsch duality can be lifted to the level of
chains. Let $\mathbb C[x,y]_0$ be the quotient of
$\C[x_1,\ldots,x_n,y_1,\ldots,y_n]$ by the ideal
$\langle x_1y_1,\ldots,x_ny_n\rangle$ and let $\bigwedge(\mathbb C^n)$
be the standard exterior representation of the Clifford algebra with
generators $e_i$, $e_i^\vee$ and relations
$e_ie_j^\vee+e_j^\vee e_i=\delta_{ij}$ for all
$i,j=1,\ldots,n$. Borisov's construction hinges on the differential
\begin{equation}\label{borisov-differential}
\delta_\infty=\sum_{i=1}^n x_i\partial_{x_i} W(x) \otimes e_i + \sum_{i=1}^n y_i\otimes e_i^\vee
\end{equation}
acting on $\mathbb C[x,y]_0\otimes \bigwedge(\mathbb C^n)$. As shown
in \cite{B2},
$(\mathbb C[x,y]_0\otimes \bigwedge(\mathbb C^n),\delta_\infty)$
contains a copy of the standard Koszul resolution of the Milnor ring
$\mathbb C[x]/dW$ in such a way that the inclusion is a
quasi-isomorphism. The starting point for this paper is to deform
$\delta_\infty$ to
\[
\delta_\pi=\sum_{i=1}^n (x_i\partial_{x_i}+\pi x_i\partial_{x_i} W(x))\otimes e_i + \sum_{i=1}^n (y_i\partial_{y_i}+\pi y_i)\otimes e_i^\vee\, ,
\]
where $\pi\in \mathbb C^*$ is an arbitrary constant. As it turns out,
the complex
$(\mathbb C[x,y]_0\otimes \bigwedge(\mathbb C^n),\delta_\pi)$ contains
a copy of the de Rham complex of the D-module
$\mathbb C[x] e^{\pi W(x)}$. The quasi-isomorphism (see e.g.\
\cite{M}) between the latter and the Milnor ring allows us to provide
an alternate chain-level realization of Berglund-H\"ubsch
duality. More precisely, our method yields a chain-level proof of the
``total unprojected'' (in the terminology of \cite{Kra})
Berglund-H\"ubsch duality, from which the usual ``projected" duality
of \cite{B2} easily follows by restriction to the invariant sectors as
in \cite{Kra}.

The key difference between our construction and \cite{B2} emerges if
one replaces $\C[x]$ with the ring
$\mathbb C_p^\dagger\langle x\rangle$ of $p$-adic overconvergent power
series. While the de Rham cohomology of the D-module
$\mathbb C_p^\dagger\langle x\rangle e^{\pi W(x)}$ (where now $\pi$ is
a fixed $(p-1)$-th root of $-p$) is still isomorphic to the $p$-adic
Milnor ring, the de Rham chain model has extra structure: a
non-trivial Frobenius endomorphism which descends to cohomology. In
this paper we show that the Frobenius endomorphisms extends naturally
to a chain map ${\rm Fr}$ acting on the full chain complex
$\mathbb C_p^\dagger\langle x,y \rangle_0\otimes \bigwedge(\mathbb
C_p^n)$.
It is then natural to ask how the Frobenius endomorphism interacts
with the Berglund-H\"ubsch duality quasi-isomorphism $\Delta$. Our
main result is that, at the level of cohomology, $\Delta$ and
${\rm Fr}$ commute up to an explicit diagonal operator whose entries
are non-negative integer powers of $p$.

The interplay between the cohomological Frobenius and
Berglund-H\"ubsch duality was first noticed in \cite{P} and used to
explore some arithmetic consequences of Berglund-H\"ubsch duality in
the spirit of \cite{W}. The present work originated as an attempt to
understand the results of \cite{P} at the level of chains. We hope
further investigate the arithmetic implications of our construction in
future work.

This paper is organized as follows. In Section 2 we review some basic
facts about invertible polynomials $W_A(x)$ over a field $\F$ defined
by a matrix $A$. In Section 3 and Section 4 we introduce our ``de
Rham'' modification $\mathcal B_A(\F)$ of Borisov's complex. In
Section 5 we point out that $\mathcal B_A(\F)$ is the total complex of
a $\Z\times \Z$-bigraded bicomplex. In Section 6 we show that
$\mathcal B_A(\F)$ is quasi isomorphic to the de Rham cohomology of a
certain D-module. To do this we follow the analogous argument given by
Borisov in \cite{B2}. However, the bigrading of \cite{B2} is no longer
preserved by our differentials and this is why we need the bigrading
introduced in Section 3 instead. In Section 7 we prove that
$\mathcal B_A(\F)$ is quasi-isomorphic to a subcomplex
$\mathcal C_A(\F)$ which is in turn canonically isomorphic to
$\mathcal C_{A^T}(\F)$. Together with the results of Section 5, this
proves unprojected Berglund-H\"ubsch duality. In Sections 8 and 9 we
specialize to the $p$-adic case and observe that the constructions of
the previous sections can be extended by replacing polynomials with
overconvergent $p$-adic power series. While not changing cohomology,
this allows for the extra room needed in order to define a natural
chain-level Frobenius endomorphism {\rm Fr} \`a la Dwork (see e.g.\
\cite{M}, \cite{SS}) whose compatibility with Berglund-H\"ubsch
duality is then addressed. Finally, in Section 10 we illustrate our
constructions by working out two simple examples.

  \bigskip\noindent \textbf{Acknowledgments:} M.A. would like to thank
  Albert Schwarz for stimulating conversations on $p$-adic methods and
  for explaining him the results of \cite{SS} which inspired the
  D-module theoretic approach to Berglund-H\"ubsch duality of the
  present paper. The work of A.P. was supported by the Natural
  Sciences and Engineering Research Council (NSERC) of Canada through
  the Discovery Grant of Noriko Yui. A.P. thanks the support of the
  NSERC. A.P. held a visiting position at the Fields Institute during
  the preparation of this paper, and would like to thank this
  institution for its hospitality.

\section{Invertible Polynomials}

Let $\F$ be a field and consider the map
$$
W\colon  \GL_n(\Z_{\geq 0}) \to \F[x]=\F[x_1 , \ldots , x_n]
$$
defined by
$$
A \mapsto W_A(x) = \sum_{i=1}^n x^{e_i A}\, ,
$$
where $\{e_i\}_{1 \leq i \leq n}$ is the standard basis of $\Z^n$, and
for $v=(v_1,\ldots,v_n) \in \Z_{\geq 0}^n$ we write
$x^{v} = x_1^{v_1}\ldots x_n^{v_n}$. For simplicity, we assume that
$\textnormal{char}\,\F =0$ or $\textnormal{char}\,\F > \det A$. A
matrix $A \in \GL_n(\Z_{\geq 0})$ is \emph{Berglund-H\"ubsch} over
$\F$ if $W_A(x)$ is an \emph{invertible polynomial}, i.e.,
$W_A(x)$ is quasi-homogeneous and
$( \partial_1 W_A(x), \ldots , \partial_n W_A(x) )$
is a regular sequence in $\F[x]$. For each
$n \in \Z_{\geq n}$ we let
\[
\BH(\F) = \bigcup_n \BH_n(\F)\, ,
\]
where
\[
\BH_n(\F) = \left\{ A \in \GL_n(\Z_{\geq 0}) \mid A\text{ is Berglund-H\"ubsch over }\F \right\}\,.
\]

\begin{remark}Berglund-H\"ubsch matrices satisfy the following
  properties.
  \begin{enumerate} 
    \item If $A \in \BH_n(\F)$ and $B \in \BH_m(\F)$, then $A \oplus B \in \BH_{n+m}(\F)$.
    \item If
      \[
      A = \left[
        \begin{array}{c|c}
          A & B \\
          \hline 
          0 & C 
        \end{array}
      \right] \in \BH(\F),
      \]
      then $C \in \BH(\F)$. We call $A \in \BH_n(\F)$ \emph{irreducible} if it
      cannot be written as $B \oplus C$ with $B,C \in \bigcup_{m\leq n} \BH_m(\F)$.
    \item Let $\cW_n \subseteq \GL_n(\Z_{\geq 0})$ be the Weyl group. Given $S\in \cW_n$ and
      $A \in \BH_n(\F)$, then $SA, AS \in \BH_n(\F)$. Moreover,
      \[
      W_{SA}(x) = W_A(x) \quad\text{and}\quad W_{AS}(x) = W_A(x) \cdot S\, ,
      \]
      where $\cdot$ denotes the right action of $\cW_n$ on
      $\F[x]$ by permutation of the variables.
  \end{enumerate}
\end{remark}

\begin{remark}
Define the \emph{group of scaling symmetries} of $A \in \BH_n(\F)$
to be $G_A = \Z^n / \Z^n A^T$. If $\F$ contains a primitive
$(\det A)$-th root of unity $\zeta$, one can consider the
$\Z_{\geq 0}^n$-action on $\F[x]$ defined for
$\lambda \in \Z^n_{\geq 0}$ by
\begin{equation}\label{eq:G_A_action_on_monomials} 
  \lambda \cdot x^\gamma = \zeta^{\gamma \lambda^T} x^\gamma.
\end{equation}
Under this action $\lambda \cdot W_A(x) = W_A(x)$ if and only if
$\lambda A^T \in (\det A)\Z^n$, which provides a canonical
identification between $G_A$ and the stabilizer of $W_A(x)$ under the
action (\ref{eq:G_A_action_on_monomials}). Unless otherwise stated, we
represent equivalence classes in $G_A$ by their canonical
representatives in $\Z^n \cap ([0,1]^n A^T)$. Using this identification,
for each $\lambda \in G_A$ we introduce a vector $J_\lambda^\vee \in \Z^n$ defined
by
\[
(J_\lambda^\vee)_i = \cas{
  0 &\text{if } \zeta^{\lambda_i} =1\, ;\\
  1 &\text{otherwise}\, ,
}
\]
and the submatrix $A^\lambda$ of $A$ such that $W_{A^\lambda}(x)$ is
obtained from $W_A(x)$ by setting $x_i = 0$ whenever
$(J_\lambda^\vee)_i = 1$.
\end{remark}

\begin{proposition}[\cite{Kre}]\label{prop:A-classification}
  Let $A \in \BH_n(\F)$ be irreducible. Then there exists $S\in \cW_n$
  such that $W_{AS}(x)$ is in one of the following canonical forms:
  \begin{enumerate}
  \item a \emph{loop},
    \[
    x_1^{a_1}x_2 + x_2^{a_2}x_3 + \ldots + x_{n-1}^{a_{n-1}}x_n + x_n^{a_n}x_1\, ,
    \]
  \item a \emph{chain},
    \[
    x_1^{a_1}x_2 + x_2^{a_2}x_3 + \ldots + x_{n-1}^{a_{n-1}}x_n + x_n^{a_n}\, .
    \]
  \end{enumerate}
\end{proposition}

\begin{corollary}
  Let $A \in \BH_n(\F)$. Then
\begin{enumerate}
  \item $A^T \in \BH_n(\F)$,
  \item for each $\lambda \in G_A$, we have
    $A^\lambda \in \BH_{n-|J_\lambda^\vee|}(\F)$, and
  \item the matrix defined by
    \[
    A^{\orb} := \bigoplus_{\lambda \in G_A} A^\lambda
    \]
    is in $\BH_{n |G_A| - \sum | J_\lambda^\vee | }(\F)$.
  \end{enumerate}
\end{corollary}

\begin{corollary}
Let $A\in {\rm BH_n(\mathbb F)}$ and let $\beta\in \mathbb Z^n$ such that $(\beta A^{-1})_i \in \mathbb Q\setminus \Z$.
\begin{enumerate}
\item If $A$ is a chain, then $(\beta A^{-1})_j$, $(\beta A^{-T})_k\in \mathbb Q\setminus \Z$ for all $1\le j\le i\le k\le n$.
\item If $A$ is a loop, then $(\beta A^{-1})_j, (\beta A^{-1})_k\in \mathbb Q\setminus \Z$ for all $1\le j,k\le n$.
\end{enumerate}
\end{corollary}

\begin{proof}
Both statements follow from
\[
 A_{ii}^T (\beta A^{-T})_i  + (\beta A^{-T})_{i+1} = \beta_i = (\beta A^{-1})_{i-1}  + A_{ii} (\beta A^{-1})_i\,,
\]
where $i$ is considered modulo $n$ in the case of loops.
\end{proof}

\section{Exterior Operators}

Let $e_1 , \ldots , e_n$ be the standard generators of $\F^n$. We
denote by $\bigwedge (\F^n)$ the exterior algebra
$\bigwedge (\F e_1 \oplus \ldots \oplus \F e_n )$ viewed as a
representation of the Clifford algebra ${\rm Cl}_n(\F)$ with
generators $e_i$ (multiplication) and $e_i^\vee$ (contraction), and
(odd) commutators $[e_i,e_j^\vee] = \delta_{ij}$ for all
$1\leq i,j \leq n$. As an $\F$-module, $\bigwedge(\F^n)$ is generated
by monomials $e^I=e_1^{I_1}\ldots e_n^{I_n}$, where
$I=(I_1,\ldots,I_n)\in \mathbb Z_{\ge 0}^n$. In particular, $e^I=0$ if
and only if $I_i\ge 2$ for some $i$. Given $A \in \BH_n(\F)$ and
$\pi \in \F^*$, for $1 \leq i \leq n$ we also consider
\[
E_{A,i} = \pi \sum_{j=1}^n e_j A_{ji}^T 
\quad \text{and} \quad
E_{A,i}^\vee = \frac{1}{\pi} \sum_{j=1}^n e_j^\vee (A^{-1})_{ji}\,,
\]
so that
\begin{equation*}
  [E_{A,i} E_{A,j}^\vee] = \sum_{k,m} A_{ki}^T (A^{-1})_{mj} [e_k, e_m^\vee] = \sum_k A_{ik} (A^{-1})_{kj} = \delta_{ij}\, .
\end{equation*}

\begin{lemma}
  If $*^A \in {\rm GL}\left(\bigwedge(\F^n)\right)$ is defined by
  \[
  *^A(e_{i_1}, \ldots , e_{i_k}) = E_{A^T,i_1}^\vee E_{A^T,i_2}^\vee \ldots E_{A^T,i_k}^\vee
  \left(  E_{A^T,1} E_{A^T,2} \ldots E_{A^T,n} \right),
  \]
  then
  \begin{enumerate}
  \item \label{star-A-part-1} $*^A E_{A,i} = e_i^\vee *^A$, $*^A E_{A,i}^\vee = e_i *^A$, and
  \item \label{star-A-part-2} $*^{A^T}*^A$ commutes with the action of ${\rm Cl}_n(\F)$ on $\bigwedge(\F^n)$.
  \end{enumerate}
\end{lemma}

\begin{proof}
  By definition,
  \[
  *^A e_i = E_{A^T,i}^\vee *^A \quad\text{and}\quad *^A e_i^\vee = E_{A^T,i} *^A.
  \]
  Therefore,
  \begin{align*}
    *^A E_{A,i} &= *^A \pi \sum_j e_j A_{ji}^T = \pi \sum_j E_{A^T,j}^\vee A_{ji}^T *^A = \sum_{k,j} e_k^\vee(A^{-T})_{kj} A_{ji}^T *^A  e_i^\vee *^A\, .
  \end{align*}
  Similarly, $*^A E_{A,i}^\vee = e_i *^A$. This proves part (\ref{star-A-part-1}). Part (\ref{star-A-part-2}) follows from
  \[
  *^{A^T} *^A e_i = *^{A^T} E_{A^T,i}^\vee *^A = e_i *^{A^T} *^A
  \]
  and
  \[
  *^{A^T}*^Ae_i^\vee = *^{A^T} E_{A,i}*^A = e_i^\vee *^{A^T} *^A\,.
  \]
\end{proof}

\begin{remark}
  The operator
  \[
  \ext = \sum_{i=1}^n e_ie_i^\vee = \sum_{i=1}^n E_{A,i}E_{A,i}^\vee
  \]
  is diagonal on $\bigwedge(\F^n)$. If $\textnormal{char}\, \F =0$, its eigenvalues count the
  total exterior degree. Moreover,
  \[
  *^A\,\ext = \sum_{i=1}^n e_i^\vee e_i *^A = (n\,\Id - \ext) *^A.
  \]
\end{remark}

\section{The Basic Complex}
Given a graded vector space $V$ endowed with a differential $d$ of
degree $1$, we denote by $(V,d)$ the corresponding chain complex and
by $H(V,d)$ its cohomology. If $V$ is bigraded and $d,d'$ are graded
commutative differentials of bidegree $(1,0)$ and $(0,1)$
respectively, we denote the corresponding bicomplex by $(V,d,d')$ and
by $H(V,d,d')$ its total cohomology. If $V$ is vector space acted upon
by a collection of commuting endomorphisms $\phi_1,\ldots,\phi_n$, we
denote the corresponding Koszul complex by
${\rm Kos}(V,\phi_1,\ldots,\phi_n)$.

Given $A \in \BH_n(\F)$, consider the subring $\widetilde{\cR_A}(\F)$
of $\F[x_1,\ldots,x_n,y_1,\ldots,y_n]$ generated by monomials
$x^\gamma y^\lambda$ such that $(\lambda A^{-T})_i \geq 0$ for all $1 \leq i \leq n$. We define
$\cR_A(\F)$ to be the quotient of $\widetilde{\cR_A}(\F)$ by the ideal
generated by monomials $x^\gamma y^\lambda$ for which
$\gamma A^{-1} \lambda^T > 0$. Given $\pi \in \F^*$, we define
$\theta_{A,i}, T_{A,i}^\vee,\psi_{A,i}^\vee,\varphi_{A,i}, \in
\End_\F \left( \cR_A(\F) \right)$ by the formulas
\begin{align*}
\theta_{A,i}(x^\gamma y^\lambda) &= \gamma_i\, x^\gamma y^\lambda\, ; \\
T_{A,i}^\vee(x^\gamma y^\lambda) & = \pi^{-1}(\lambda A^{-T})_i\, x^\gamma y^\lambda\, ;\\
\psi_{A,i}^\vee(x^\gamma y^\lambda)& = x^\gamma y^{\lambda + e_iA^T}\, ;\\
\varphi_{A,i}(x^\gamma y^\lambda) &= \pi \left(\theta_i W_A(x) \right)x^\gamma y^\lambda =\pi\sum_{j=1}^n A_{ji}\, x^{\gamma + e_j A}y^\lambda\,. 
\end{align*}
We also define the odd linear endomorphisms of
$\cR_A(\F) \otimes \bigwedge(\F^n)$
\[
d_{A,i} = \left(  \theta_{A,i} + \varphi_{A,i} \right) e_i\,, \quad  d_A =\sum_{i=1}^n d_{A,i} 
\]
as well as
\[
d_{A,i}^\vee = \left( T_{A,i}^\vee + \psi_{A,i}^\vee \right) e_i^\vee\, ,\quad d_A^\vee = \sum_{i=1}^n d_{A,i}^\vee\, .
\]

\begin{lemma}\label{chain} $\cB_A(\F) = \left( \cR_A(\F) \otimes \bigwedge (\F^n) , d_A+d_A^\vee \right)$ is a chain complex.
\end{lemma}

\begin{proof}
  The morphism $d_A$ is the Koszul differential for the sequence
  \[
  \left( \theta_{A,1} + \varphi_{A,1} , \theta_{A,2} + \varphi_{A,2} , \ldots, \theta_{A,n} + \varphi_{A,n} \right)
  \]
  of commuting operators acting on $\cR_A(\F)$. Therefore,
  $[d_A,d_A]=0$, and similarly $[d_A^\vee,d_A^\vee]=0$. Moreover, since
  $\left( \theta_{A,i} + \varphi_{A,i} \right)$ and
  $\left( T_{A,j}^\vee + \psi_{A,j}^\vee \right)$ commute, 
    \begin{align*}
    [d_{A,i},d_{A,j}^\vee] &= \left[ \left( \theta_{A,i} + \varphi_{A,i} \right) e_i ,
                                   \left( T_{A,j}^\vee + \psi_{A,j}^\vee \right) e_j^\vee        
                           \right] \\
                         &= \left( \theta_{A,i}+\varphi_{A,i} \right) 
                            \left( T_{A,j}^\vee + \psi_{A,j}^\vee \right) 
                            \left[ e_i , e_j^\vee \right] \\
                         &= \left( \theta_{A,i}+\varphi_{A,i} \right) 
                            \left( T_{A,j}^\vee + \psi_{A,j}^\vee \right) 
                            \delta_{ij}\, .
  \end{align*}
  If
  \[0 \neq \left( \theta_{A,i}\, T_{A,i}^\vee \right)(x^\gamma y^\lambda)
  = \gamma_i\, (A^{-1}\lambda^T)_i\,x^\gamma y^\lambda\,,
  \]
  then $x^\gamma y^\lambda = 0$ in $\cR_A(\F)$ and thus $\left( \theta_{A,i} + \varphi_{A,i} \right) \left( T_{A,i}^\vee + \psi_{A,i}^\vee\right) = 0$. For
  \[
  \left( \varphi_{A,i}\, T_{A,i}^\vee \right) (x^\gamma y^\lambda)
  = \sum_{j=1}^n A_{ji} (\lambda A^{-T})_i\, x^{\gamma + e_j A} y^\lambda
  \]
  we note that if for some $j$ we have $A_{ji},(\lambda A^{-T})_i > 0$, then
  \[
    (\gamma+e_jA) A^{-1} \lambda^T > (e_j A) (A^{-1} \lambda^T) =\sum_{m=1}^n A_{jm} (A^{-1}\lambda^T)_m > A_{ji} (A^{-1}\lambda^T)_i > 0 
  \]
  and conclude as before that $x^{\gamma + e_j A} y^\lambda = 0$ in $\cR_A(\F)$. 
  It is similarly shown that $\varphi_{A,i}\, \psi_{A,i}^\vee = 0$ and
  $\theta_{A,i}\, \psi_{A,i}^\vee = 0$. Therefore,
  $[d_A,d_A^\vee] = 0$.
\end{proof}

\begin{remark}
Note that for any monomial $x^\gamma y^\lambda$, the vector $\lambda$ encodes
a group element of $G_A$ by (\ref{eq:G_A_action_on_monomials}). For
$\lambda \in G_A$ we take $0 \leq (\lambda A^{-T} )_i < 1$ for each $i$, 
so $\gamma A^{-1} \lambda^T = 0$ means that $\gamma_i = 0$ if $\lambda$ acts
non-trivially on $x_i$, that is, if $(J_\lambda^\vee)_i = 1$.
\end{remark}

\section{Bigrading}

Let $P_{A,i}^\vee \in \End_\F\left(\cR_A(\F)\right)$ be given by
\[
P_{A,i}^\vee(x^\gamma y^\lambda) 
= \cas{
  0,                &\text{if }(\lambda A^{-T})_i =0\, ; \\
  x^\gamma y^\lambda, &\text{otherwise}\, .
}
\]

\begin{lemma}
  Let $Q_{A,i} , Q_{A,i}^\vee, Q_A$ and $Q_A^\vee$ be linear
  endomorphisms of $\cR_A(\F) \otimes \bigwedge(\F^n)$ defined by
  \[
  Q_{A,i}^\vee = P_{A,i}^\vee e_i^\vee e_i
  \quad\text{and}\quad
  Q_{A,i} = e_i e_i^\vee + Q_{A,i}^\vee
  \]
  as well as
  \[
  Q_A = \sum_{i=1}^n Q_{A,i}
  \quad\text{and}\quad
  Q_A^\vee = \sum_{i=1}^n Q_{A,i}^\vee\, .
  \]
  Then for each $1 \leq i,j \leq n$,
  \begin{enumerate}
    \item \label{lem:grading-comm-1} $[Q_{A,i}, Q_{A,j}^\vee] = 0$,
    \item \label{lem:grading-comm-2} $[Q_{A,i}^\vee, d_{A,j}] = 0$ and $[Q_{A,i}, d_{A,j}] = \delta_{ij}d_{A,j}$,
    \item \label{lem:grading-comm-3} $[Q_{A,i}, d_{A,j}^\vee] = 0$ and $[Q_{A,i}^\vee, d_{A,j}^\vee] = \delta_{ij} d_{A,j}^\vee$.
  \end{enumerate}
\end{lemma}

\begin{proof}
  The operators $Q_{A,i}$ and $Q_{A,j}^\vee$ commute because they
  have monomials of the form $ x^\gamma y^\lambda e^I $ as a common basis of eigenvectors,
  which proves (\ref{lem:grading-comm-1}). For (\ref{lem:grading-comm-2})
  \begin{align*}
    [Q_{A,i}^\vee,d_{A,j}] &= \left[ P_{A,i}^\vee\, e_i^\vee e_i, 
                                   \left(\theta_{A,j} + \varphi_{A,j}\right) e_j
                          \right] \\
                         &= P_{A,i}^\vee \left(\theta_{A,j} + \varphi_{A,j} \right)
                            [e_i^\vee e_i, e_j] \\
                         &= \delta_{ij} P_{A,i}^\vee 
                            \left( \theta_{A,j} + \varphi_{A,j} \right) e_j^\vee.
  \end{align*}
  The proof of Lemma (\ref{chain}) shows that
  $P_{A,i}^\vee\, \varphi_{A,i} = 0$. Similarly,
  $P_{A,i}^\vee\, \theta_{A,i}(x^\gamma y^\lambda) \neq 0$ implies
  that $\gamma_i (A^{-1} \lambda^T)_i > 0$ so that the corresponding
  term is $0$ in $\cR_A(\F)$. Therefore,
  $[Q_{A,i}^\vee,d_{A,j}] = 0$, which in turn implies that
  \begin{align*}
    [Q_{A_i},d_{A,j}] = \left(\theta_{A,j} + \varphi_{A,j} \right)[e_ie_i^\vee, e_j] = \delta_{ij} d_{A,j}\, .
  \end{align*}
For part (\ref{lem:grading-comm-3}), if $i\neq j$
  \[
  [P_{A,i}^\vee, T_{A,j}^\vee + \psi_{A,j}^\vee] = 0 =[e_i^\vee e_i, e_j^\vee]\, ,
  \]
  and otherwise $e_i^\vee e_i^\vee e_i = 0$ and
  $e_i^\vee e_i e_i^\vee = e_i^\vee$, which means that
  \begin{align*}
    [Q_{A,i}^\vee, d_{A,j}^\vee] = \delta_{ij} P_{A,j}^\vee 
                                  \left(T_{A,j}^\vee + \psi_{A,j}^\vee\right) e_i^\vee = \delta_{ij} d_{A,j}^\vee\, .
  \end{align*}
  Similarly,
  \begin{align*}
    [Q_{A,i}, d_{A,j}^\vee] &= \left[e_i e_i^\vee , 
                              \left( T_{A,j}^\vee + \psi_{A,j}^\vee \right) e_j^\vee 
                             \right] + \delta_{ij} d_{A,j}^\vee\\
                          &= \delta_{ij} \left( - \left(T_{A,j}^\vee + \psi_{A,j}^\vee\right)
                                               e_j^\vee e_j e_j^\vee + d_{A,j}^\vee
                                        \right) = 0\, .
  \end{align*}
\end{proof}

\begin{remark}
  In particular, with respect to the $\text{Spec}(Q_A) \times \text{Spec}(Q_A^\vee)$
  bigrading, $\cB_A(\F)$ is the total complex of the bicomplex
  $\left( \cR_A(\F) \otimes \bigwedge(\F^n) , d_A, d_A^\vee \right)$.
 
\end{remark}

\section{Unprojected Orbifold de Rham Cohomology}

Given $A \in \BH_n(\F)$ and $\pi \in \F^*$, let 
\[
M_A(\F) = \F[x]e^{\pi W_A(x)}
\]
be the module over the Weyl algebra
$\F[x_1,\ldots,x_n,\partial_1,\ldots,\partial_n]$ on which $x_i$ acts
by multiplication and
\begin{equation}\label{eq:partial_i-action}
\partial_i \cdot P(x) = \partial_iP(x) + \pi(\partial_i W_A(x))P(x)
\end{equation}
for each $1 \leq i \leq n$ and $P(x) \in
M_A(\F)$. We denote by
$\DR_A(\F)$ the \emph{de Rham complex} of $M_A(\F)$, which is by
definition the Koszul complex
\[
\textnormal{Kos}\left( M_A(\F);\,\partial_1,\partial_2, \ldots, \partial_n \right),
\]
where each $\partial_i$ acts as in equation
(\ref{eq:partial_i-action}). Given $\lambda \in \Z_{\geq 0}^n$ such
that $(\lambda A^{-T})_i \geq 0$ for all $1\leq i \leq n$, let
$\cR_A^\lambda(\F) \subseteq \cR_A(\F)$ be generated by monomials of
the form $x^\gamma y^{\lambda + \mu A^T}$ for some
$\gamma, \mu \in \Z_{\geq 0}^n$. Then
$\cR_A^\lambda(\F)\otimes\bigwedge(\F^n)$ is closed under
$d_A + d_A^\vee$ and we denote by
$\cB_A^\lambda(\F) \subseteq \cB_A(\F)$ the corresponding subcomplex.

\begin{lemma}\label{lem:B_A-decomposition}
If $A \in \BH_n(\F)$, then
  \begin{enumerate}
  \item $\cB_A(\F) \cong \bigoplus_{\lambda \in G_A} \cB_A^\lambda(\F)$, and
  \item $\cB_A^\lambda(\F)$ is quasi-isomorphic to $\cB_{A^\lambda}^{0}(\F)$.
  \end{enumerate}
\end{lemma}

\begin{proof}
  Part (1) holds because $G_A = \Z^n / \Z^n A^T$ and
  $(\lambda A^{-T})_i \geq 0$ for $1 \leq i \leq n$. For (2),
  note that since
  $y^\lambda x_i = 0$ in $\cR_A^\lambda(\F)$ for any $i$ such that
  $(J_\lambda^\vee)_i = 1$
  \[
  \cR_A^\lambda(\F) \cong \cR_{A^\lambda}^{0}(\F) 
  \otimes
  \F[ \psi_{A,i}^\vee ]_{(J_\lambda^\vee)_i=1}\,y^\lambda\, ,
  \]
  and thus
  \begin{equation}\label{eq:cB_A^lambda-decomposition}
  \cB_A^\lambda(\F) \cong \cB_{A^\lambda}^{0}(\F) \otimes {\rm Kos}(\F[ \psi_{A,i}^\vee ]_{(J_\lambda^\vee)_i=1}\,y^\lambda,(T_{A,i}^\vee + \psi_{A,i}^\vee)_{(J_\lambda)_{i}=1})\, .
  \end{equation}
  The cohomology of the second factor is isomorphic to $\F y^\lambda$,
  making the inclusion
  $\cB_{A^\lambda}^{0}(\F) \hookrightarrow \cB_A^\lambda(\F)$ a
  quasi-isomorphism.
\end{proof}

\begin{proposition}\label{lem:B_A^0-to-de-Rham}
  The complex $\cB_A^{0}(\F)$ is canonically quasi-isomorphic to $\DR_A(\F)$.
\end{proposition}

\begin{proof}
  The map
  $\Theta \colon M_A(\F) \otimes \bigwedge (\F^n) \to \cR_A^{0}(\F)
  \otimes \bigwedge(\F^n)$
  defined by $\Theta(x^\gamma e^I) = x^{\gamma + I} e^I$ gives rise to
  an embedding
  \[
  \DR_A(\F) \hookrightarrow \cB_A^{0}(\F)
  \]
  of complexes. For each $\gamma \in \Z_{\geq 0}^n$, let
  $\bigwedge_\gamma = \bigwedge
  \left(\bigoplus_{\gamma_i = 0} \F e_i \right)$. Then
  \[
  \left( \cR_A^{0}(\F) \otimes \bigwedge(\F^n), d_A^\vee \right) = \bigoplus_{\gamma,I} \cC_{\gamma,I}\, ,
  \]
  where
  \[
  \cC_{\gamma,I} = \left(  x^{\gamma+I} \F[ y^{e_i A^T} ]_{(\gamma + I)_i = 0}
  \otimes e^I\bigwedge\nolimits_{\gamma+I}\,,
  \sum_{(\gamma + I)_i =0} d_{A,i}^\vee \right)
  \]
  is a Koszul complex with cohomology $\F x^{\gamma + I} e^I$. This
  implies that
  \[
  H\left( \frac{\cR_A^{0}(\F) \otimes \bigwedge(\F^n)}{\textnormal{Im}\,\Theta}, d_A^\vee \right)=0\, ,
  \]
  and using the spectral sequence of first quadrant bicomplexes we conclude that
  \[
  H\left( \cB_A^{0}(\F) / \DR_A(\F) \right) = 0\, .
  \]
  Therefore, the inclusion $\DR_A(\F) \hookrightarrow \cB_A^{0}(\F)$ is
  a quasi-isomorphism.
\end{proof}

\begin{corollary}\label{cor:BH_n-basis}
  Let $A\in {\rm BH}_n(\mathbb F)$ and let $S(A)$ be the collection of
  monomials $x^\gamma e^I$ such that $|I|=1$ and
  $1\le \gamma_i\le a_i=A_{ii}$ for all $i=1,\ldots,n$. Then
  \begin{enumerate}
  \item $H(\mathcal B_A^0(\F))$ is isomorphic to the Milnor ring $\mathbb F[x]/dW_A(x)$.
  \item If $W_A(x)$ is a loop, then $H(\mathcal B_A^0(\F))$ is generated by monomials in $S(A)$.
  \item If $W_A(x)$ is a chain, then $H(\mathcal B_A^0(\F))$ is generated
         by those monomials in $S(A)$ of the form 
    \[
    x_1^{a_1}x_2x_3^{a_2}x_4\ldots x^{a_{2m-1}}_{2m-1} x_{2m} x^{\gamma_{2m+1}}_{2m+1}\ldots x_n^{\gamma_n}
    \]
    where $m\ge 0$ is such that $\gamma_{2m+1}<a_{2m+1}$.
\end{enumerate}
\end{corollary}
\begin{proof}
  Part (1) follows from Proposition \ref{lem:B_A^0-to-de-Rham} and the
  fact (see e.g.\ \cite{M}) that there is a linear map from
  $\F[x]/dW_A(x)$ to $H(\DR_A(\mathbb F))$ sending monomials to
  monomials. Comparison with the standard monomial basis for the
  Milnor ring of chains and loops (see e.g.\ \cite{Kre}) establishes
  (2) and (3).
\end{proof}

\begin{proposition}\label{prop:orbifold-to-B_A}
  The natural inclusion of 
  $\DR_{A^\textnormal{orb}}(\F)$ into $\cB_A(\F)$ is a
  quasi-isomorphism.
\end{proposition}

\begin{proof}
  The proposition follows from Lemma~\ref{lem:B_A-decomposition} and Proposition
  \ref{lem:B_A^0-to-de-Rham}.
\end{proof}

\section{Unprojected Duality}

Given $A \in \BH_n(\F)$ and $\pi \in \F^*$, let
$\psi_{A,i}, T_{A,i} \in \End_\F(\cR_A(\F))$ for $1 \leq i \leq n$ be
defined by
\[
\psi_{A,i}(x^\gamma y^\lambda) = x^{\gamma + e_i A} y^\lambda 
\]
and
\[
T_{A,i} (x^\gamma y^\lambda) = \pi^{-1} (\gamma A^{-1})_i\,x^\gamma y^\lambda,
\]
so that $d_A=\sum_{i=1}^n \hat d_{A,i}$, where
\[
\hat d_{A,i} = (T_{A,i} + \psi_{A,i}) E_{A,i}\, .
\]

\begin{remark}
Since we are using logarithmic differentials, $e_i$ can be naturally interpreted
as $dx_i / x_i$. One motivation for the change of basis to $E_{A,i}$ is the Shioda map $x^\gamma \mapsto z^{\gamma A^{-1} \det(A)}$ which sends $W_A$ to 
$z^{e_1 \det(A)} + \ldots + z^{e_n \det(A)}$. If we interpret $E_{A,i}$ as 
$dz_i / z_i$, its definition is simply the chain rule. 
\end{remark}

Let $\cS_A(\F) \subseteq \cR_A(\F)$ be generated by monomials
$x^\gamma y^\lambda$ such that $(\gamma A^{-1})_i \geq 0$ for all
$1 \leq i \leq n$. Then $\cS_A(\F) \otimes \bigwedge(\F^n)$ is closed
under $d_A + d_A^\vee$. Let $\cC_A(\F) \subseteq \cB_A(\F)$ denote the
corresponding subcomplex.

\begin{lemma}\label{lem:restriction-quasi-isomorphism}
The inclusion $\cC_A(\F) \hookrightarrow \cB_A(\F)$ is a quasi-isomorphism.
\end{lemma}

\begin{proof}
  Consider the filtration
  \[
  \cS_A(\F) \subseteq F^n \subseteq F^{n-1} \subseteq \ldots \subseteq F^1 = \cR_A(\F)\, ,
  \]
  where $F^i$ is spanned by monomials $x^\gamma y^\lambda$ such that
  $(\gamma A^{-1})_j \geq 0$ for all $j < i$. In particular,
  $F^i / F^{i+1}$ is canonically identified with the space of
  monomials $x^\gamma y^\lambda$ such that $(\gamma A^{-1})_i < 0$.
  Consider the filtration
  $G^\bullet(\F) = F^\bullet \otimes \bigwedge(\F^n)$ of
  $\cR_A(\F) \otimes \bigwedge(\F^n)$. Notice that
  \[
  \left( \frac{G^i(\F))}{G^{i+1}(\F)};\,d_A,d_A^\vee \right)
  \]
  is a bicomplex with respect to the
  $\textnormal{Spec}(Q_A) \times \textnormal{Spec}(Q_A^\vee)$
  bigrading, while
  \[
  \left( \frac{G^i(\F)}{G^{i+1}(\F)};\,\hat d_{A,i}, d_A - \hat d_{A,i} \right)
  \]
  is a bicomplex with respect to the
  \[
  \textnormal{Spec}(E_{A,i}\, E_{A,i}^\vee) \times \textnormal( \ext - E_{A,i}\,E_{A,i}^\vee)
  \]
  bigrading. Therefore, in order to prove that
  \begin{equation}\label{eq:G^i-quot-coh-d+d}
  H\left( \frac{G^i(\F)}{G^{i+1}(\F)},d_A + d_A^\vee \right) = 0\, ,
  \end{equation}
  it is sufficient to show that
  \begin{equation}\label{eq:G^i-quot-coh-hat-d}
    H \left( \frac{G^i(\F)}{G^{i+1}(\F)},\hat d_{A,i} \right)  = 0\, .
  \end{equation}
  If this is the case, the result then follows from the spectral
  sequence of the filtered complex
  $\left( \cB_A(\F), G^\bullet(\F)\right)$. To prove
  (\ref{eq:G^i-quot-coh-hat-d}), we distinguish two cases:

  First, suppose that $\textnormal{char}\,\F = 0$. In this case,
  $T_{A,i}$ acts by nonzero eigenvalues on $F^i / F^{i+1}$.  By
  looking at the filtration of $F^i / F^{i+1}$ by
  $\textnormal{Spec}(T_{A,i})$, we conclude that
  $T_{A,i} + \psi_{A,i}$ is injective. Therefore,
  $H \left( \frac{G^i(\F)}{G^{i+1}(\F)},\hat d_{A,i} \right)$
  is concentrated in top
  ${\rm Spec}(E_{A,i}E_{A,i}^\vee)$-degree and isomorphic to the quotient
  \[
  F^i \big/ \left( F^{i+1} + \textnormal{Im}(T_{A,i} + \psi_{A,i}) \right).
  \]
  On the other hand, for each $f \in F^i$ there exists $N \in \N$
  such that $\psi_{A,i}^n f \in F^{i+1}$, which implies
  (\ref{eq:G^i-quot-coh-hat-d}). See Figure~\ref{fig: Filtration by Spec T_AT for W_AT}
  for an illustration.

  Second, suppose that $\textnormal{char}\,\F = p > \det A$. Let
  $\K$ be a field such that $\textnormal{char}\,\K = 0$ and
  $\F = \K/p\K$. Consider the short exact sequence of complexes
  \[
  \begin{CD}
    0\! 
    \!@>>>
    \!\left( \frac{G^i(\K)}{G^{i+1}(\K)}, \hat d_{A,i} \right)
    @>p>>
    \!\left( \frac{G^i(\K)}{G^{i+1}(\K)}, \hat d_{A,i} \right)
    \!@>>>
    \!\left( \frac{G^i(\K)}{G^{i+1}(\K)}, \hat d_{A,i} \right) 
    \!@>>>
    \!0\, .
  \end{CD}
  \]
  Taking the long exact sequence and using the characteristic $0$ case
  established above yields (\ref{eq:G^i-quot-coh-hat-d}).
\end{proof}

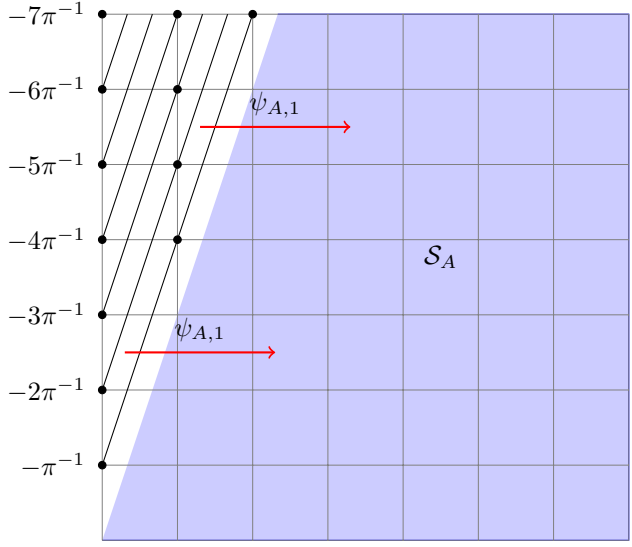
\begin{figure}
\begin{tikzpicture}
\draw [blue,fill=blue!20] (0,0) -- (7,0) -- (7,7) -- (2.3333333,7);
\draw[help lines] (0,0) grid (7,7);
\draw (0,1) -- (2,7);
\draw[fill] (0,1) circle [radius=0.05]; \node [left] at (-0.1,1) {$- \pi^{-1}$};
\draw[fill] (1,4) circle [radius=0.05];
\draw[fill] (2,7) circle [radius=0.05];
\draw (0,2) -- (1.6666666,7);
\draw[fill] (0,2) circle [radius=0.05]; \node [left] at (-0.1,2) {$-2\pi^{-1}$};
\draw[fill] (1,5) circle [radius=0.05];
\draw (0,3) -- (1.3333333,7);
\draw[fill] (0,3) circle [radius=0.05]; \node [left] at (-0.1,3) {$-3\pi^{-1}$};
\draw[fill] (1,6) circle [radius=0.05];
\draw (0,4) -- (1,7);
\draw[fill] (0,4) circle [radius=0.05]; \node [left] at (-0.1,4) {$-4\pi^{-1}$};
\draw[fill] (1,7) circle [radius=0.05];
\draw (0,5) -- (0.6666666,7);
\draw[fill] (0,5) circle [radius=0.05]; \node [left] at (-0.1,5) {$-5\pi^{-1}$};
\draw (0,6) -- (0.3333333,7); 
\draw[fill] (0,6) circle [radius=0.05]; \node [left] at (-0.1,6) {$-6\pi^{-1}$};
\draw[fill] (0,7) circle [radius=0.05]; \node [left] at (-0.1,7) {$-7\pi^{-1}$};
\draw [->, red, thick] (0.3,2.5) -- (2.3,2.5); \node[above] at (1.3,2.5) {$\psi_{A,1}$};
\draw [->, red, thick] (1.3,5.5) -- (3.3,5.5); \node[above] at (2.3,5.5) {$\psi_{A,1}$};
\node[above] at (4.5,3.5) {$\cS_A$};
\end{tikzpicture}
\caption{Eigenspaces of $T_{A,2}$ in $F^2$ for $W_A = x_1^2 + x_1 x_2^3$, with eigenvalues designated along the left. Each point represents $\gamma$ in $x^\gamma$.}
\label{fig: Filtration by Spec T_AT for W_AT}
\end{figure}

\begin{proposition}\label{prop:delta-A}
  Let $D^A\colon \cS_A(\F) \to \cS_{A^T}(\F)$ be defined by
  $D^A(x^\gamma y^\lambda) = x^\lambda y^\gamma$. Then,
  \[
  \Delta^A = D^A \otimes *^A \colon \cS_A(\F) \otimes \bigwedge(\F^n)
  \to \cS_{A^T}(\F) \otimes \bigwedge(\F^n)
  \]
  induces an isomorphism of complexes
  \[
  \begin{CD}
    \cC_A(\F) @>\cong>> \cC_{A^T}(\F)\, .
  \end{CD}
  \]
\end{proposition}

\begin{proof}
  Since $D^{A^T} D^{A} = \Id$ and
  $*^A\in{\rm GL}\left(\bigwedge(\F^n)\right)$, we only need to prove
  that $\Delta^A$ is a chain map. Using
  \[
  D^A \psi_{A,i} = \psi_{A^T,i}^\vee D^A
  \quad\text{and}\quad
  D^A T_{A,i} = T_{A^T,i}^\vee D^A\, ,
  \]
  it follows that
  \[
  \Delta^A\left( (T_{A,i} + \psi_{A,i} ) E_{A,i} \right) =
  \left( (T_{A^T,i}^\vee + \psi_{A^T,i}^\vee) e_i^\vee \right) \Delta^A
  \]
  and
  \[
  \Delta^A \left( (T_{A,i}^\vee + \psi_{A,i}^\vee)e_i^\vee \right) =
  \left( (T_{A^T,i} + \psi_{A^T,i}) E_{A^T,i} \right) \Delta^A.
  \]
\end{proof}

\begin{theorem}[Unprojected Berglund-H\"ubsch Duality]
  The complexes $\DR_{A^\textnormal{orb}}(\F)$ and
  $\DR_{(A^T)^\textnormal{orb}}(\F)$ are canonically quasi-isomorphic.
\end{theorem}

\begin{proof}
  The theorem follows from Proposition~\ref{prop:delta-A},
  Lemma~\ref{lem:restriction-quasi-isomorphism}, and
  Proposition~\ref{prop:orbifold-to-B_A}.
\end{proof}

\section{Overconvergent Power Series}

Let $p \in \Z_{\geq 0}$ be a prime, $\K = \C_p$, $\F = \K / p\K$,
$A \in \BH_n(\F)$, and $\pi \in \K$ such that $\pi^{p-1} = -p$. Let
$\widetilde{\cR_A}^\dagger(\K)$ be the ring of
\emph{overconvergent power series}
\[
\sum_{\gamma, \lambda \in \Z_{\geq 0}^n} a_{\gamma, \lambda}\,x^\gamma y^\lambda
\]
such that $(\lambda A^{-T})_i \geq 0$ for all $1 \leq i < n$, and
such that there exists $M > 0$ for which 
\begin{equation}\label{overconvergent}
\ord_p( a_{\gamma, \lambda} ) \geq M ( |\gamma| + |\lambda| )
\end{equation}
for all but finitely many $\gamma, \lambda$.  Similarly, define
$\cR_A^\dagger(\K)$, $\cS_A^\dagger(\K)$, $\cB_A^\dagger(\K)$, and
$\cC_A^\dagger(\K)$ as before, by replacing polynomials with
overconvergent power series.

\begin{lemma}\,
  \begin{itemize}
  \item \label{lem:overc-poly} The inclusions
    \[
    \begin{CD}
      \cB_A(\K) @>>> \cB_A^\dagger(\K) \\
      @AAA            @AAA \\
      \cC_A(\K) @>>> \cC_A^\dagger(\K)
    \end{CD}
    \]
    are quasi-isomorphisms.
  \item \label{lem:delta-A-overc} $\Delta^A$ extends to an isomorphism of complexes
    \[
    \cC_A^\dagger(\K) \xrightarrow{\cong} \cC_{A^T}^\dagger(\K).
    \]
  \end{itemize}
\end{lemma}

\begin{proof}
  To prove (1), let
  $f = \sum a_{\gamma, \lambda}\, x^\gamma y^\lambda$ be an
  overconvergent power series. Since
  $\ord_p(a_{\gamma,\lambda}) \geq 1$ for all but finitely many
  $\gamma, \lambda$, by reducing modulo $p$ we obtain a polynomial
  $\overline{f} \in \F[x_1,\ldots,x_n,y_1,\ldots,y_n]$. Therefore,
  $\cR_A(\K)$ and $\cR_A^\dagger(\K)$ both reduce modulo $p$ to
  $\cR_A(\F)$. Since $\cB_A(\F)$ and $\cC_A(\F)$ decompose into
  subcomplexes with cohomology concentrated in top degree, the
  statement follows from the long exact sequences of the diagram
  \[
  \begin{CD}
    0 @>>> \cB_A^\dagger(\K) @>{p}>> \cB_A^\dagger(\K) @>>> \cB_A(\F) @>>> 0 \\
    @.      @AAA                    @AAA                  @A{\cong}AA   @. \\
    0 @>>> \cB_A(\K)        @>{p}>> \cB_A(\K)         @>>> \cB_A(\F) @>>> 0 \\
    @.      @AAA                    @AAA                  @AAA          @. \\
      @.    0               @.        0               @.   0
  \end{CD}
  \]
  as in \cite[Theorem 8.5]{M}. Therefore
  $\cB_A(\K) \hookrightarrow \cB_A^\dagger(\K)$ is
  quasi-isomorphism. Similarly,
  $\cC_A(\K) \hookrightarrow \cC_A^\dagger(\K)$ is a
  quasi-isomorphism. The rest of the statement follows from
  Lemma~\ref{lem:restriction-quasi-isomorphism}.

  Part (2) follows as in Proposition~\ref{prop:delta-A} after noticing
  that the overconvergence property is preserved by $D^A$.
\end{proof}

\section{The Frobenius Endomorphism}

Let
$p^{Q_A}, p^{Q_A^\vee} \in \End \left(\cR_A^\dagger(\K) \otimes
  \bigwedge(\K^n)\right)$ be defined by
\[
p^{Q_A}(x^\gamma y^\lambda e^I) = p^\xi x^\gamma y^\lambda e^I
\quad\text{and}\quad
p^{Q_A^\vee}(x^\gamma y^\lambda e^I) = p^{\xi^\vee} x^\gamma y^\lambda e^I,
\]
where $Q_A(x^\gamma y^\lambda e^I) = \xi x^\gamma y^\lambda e^I$ and
$Q_A^\vee(x^\gamma y^\lambda e^I) = \xi^\vee x^\gamma y^\lambda
e^I$.

\begin{lemma}
  If $\Theta_A' \colon \cR_A^\dagger(\K) \to \cR_{pA}^\dagger(\K)$ is
  defined by $\Theta_A'(x^\gamma y^\lambda) = x^{p\gamma} y^{p\lambda}$, then
  \[
  \Fr_A' = \left( \Theta_A' \otimes \Id_{\bigwedge(\K^n)} \right) p^{Q_A}
  \] 
  defines a chain map $\cB_A^\dagger(\K) \to \cB_{pA}^\dagger(\K)$.
\end{lemma}

\begin{proof}
  This follows from
  \begin{align*}
    \Fr_A' d_{A,i} & = \Theta_A' (\theta_{A,i} + \varphi_{A,i})e_i\, p^{Q_A + 1} 
                  = d_{pA,i} \Fr_A'\,;\\
    \Fr_A' d_{A,i}^\vee & = \Theta_A' (T_{A,i}^\vee + \psi_{A,i}^\vee) e_i^\vee\, p^{Q_A}
                       = d_{pA,i}^\vee \Fr_A'\,.
\end{align*}
\end{proof}

\begin{lemma}
  If $\Theta_A'' \colon \cR_{pA}^\dagger(\K) \to \cR_A^\dagger(\K)$
  is defined by
  \[
  \Theta_A'' (x^\gamma y^\lambda) = Z_A(x) Z_{A^T}(y) x^\gamma y^\lambda,
  \]
  where $Z_A(x) = e^{\pi\left(W_{pA}(x)-W_A(x) \right)}$, then
  \[
  \Fr_A'' = \left(\Theta_A'' \otimes \Id_{\bigwedge(\K^n)} \right) p^{Q_{pA}^\vee}
  \]
  defines a chain map $\cB_{pA}^\dagger(\K) \to \cB_A^\dagger(\K)$.
\end{lemma}

\begin{proof}
  It is well known (see e.g.\ \cite{M}) that $Z_A(x)$ satisfies
  (\ref{overconvergent}). Therefore, $\Theta_A''$ is well defined. We
  compute
  \begin{align*}
    \left( \theta_{A,i} + \varphi_{A,i}\right) \Theta_A''
    &= \theta_{A,i} \left(W_{pA}(x) - W_A(x) \right) \Theta_A''
       +\Theta_A''\, \theta_{pA,i} + \varphi_{A,i}\, \Theta_A'' \\
    &= \Theta_A''\, \varphi_{pA,i} - \varphi_{A,i}\, \Theta_A'' 
    + \Theta_A''\, \theta_{pA,i} + \varphi_{A,i}\, \Theta_A'' \\
    &= \Theta_A''\left( \theta_{pA,i} + \varphi_{pA,i} \right),
  \end{align*}
  from which we see that
  \[
  \Fr_A''\, d_{pA,i} = \Theta_A'' \left(\theta_{pA,i} + \varphi_{pA,i}\right) e_i\, p^{Q_{pA}^\vee}
                  = d_{A,i}\, \Fr_A''\, .
  \]
  On the other hand, for each $1 \leq i \leq n$
  \begin{align*} 
    \left(T_{A,i}^\vee + \psi_{A,i}^\vee\right) \Theta_{A,p} 
    &=  T_{A,i}^\vee \left( W_{pA^T}(y) - W_{A^T} (y) \right) \Theta_A''
        + \Theta_A''\, T_{A,i}^\vee + \psi_{A,i}^\vee\, \Theta_A'' \\
    &= p\, \Theta_A''\, \psi_{pA,i}^\vee - \psi_{A,i}^\vee\, \Theta_A'' + p\,\theta_A''\, T_{pA,i}''
    + \psi_{A,i}^\vee\, \Theta_{A}'' \\
    &= p\, \Theta_A''\left( T_{pA,i}^\vee + \psi_{pA,i}^\vee\right)\, .
  \end{align*}
  Therefore,
  \[
  \Fr_A''\, d_{pA,i}^\vee = \Theta_A'' \left(T_{pA,i}^\vee + \psi_{pA,i}^\vee\right) e_i\, p^{Q_A^\vee -1} 
                        = d_{A,i}^\vee\, \Fr_A''.
  \]
\end{proof}

\begin{lemma}
  Let $\widehat{P}_{A,i} \in \End_\K \left(\cS_A^\dagger(\K)\right)$ be defined by
  \[
  \widehat{P}_{A,i}(x^\gamma y^\lambda) = 
  \cas{
    0  &\textnormal{if }(\gamma A^{-1})_i =0\, ;\\
    x^\gamma y^\lambda   &\textnormal{otherwise}\, .
  }
  \]
  \begin{enumerate}
   \item If we define
   \begin{align*}
  \widehat{Q}_{A,i} & = \widehat{P}_{A,i} \, E_{A,i} \, E_{A,i}^\vee\\
  \widehat{Q}_{A,i}^\vee & = E_{A,i}^\vee\, E_{A,i} + \widehat{Q}_{A,i}
  \end{align*} 
  then
    \[
    \Delta^A\, Q_{A,i}^\vee = \widehat{Q}_{A^T,i} \, \Delta^A
    \quad\text{and}\quad
    \Delta^A \, Q_{A,i} = \widehat{Q}_{A^T,i}^\vee \, \Delta^A\, .
    \]
  \item If we define
  \[
  \widehat{d}_{A,i} = \left(T_{A,i} + \psi_{A,i} \right) E_{A,i}\quad \text{and}\quad
   \widehat{D}_{A,i}^\vee = \left(\theta_{A,i}^\vee + \psi_{A,i}^\vee \right) E_{A,i}^\vee
   \] 
   then
     \[
      [ \widehat{Q}_{A,i}, \widehat{d}_{A,j}^\vee ] = 0
      = [ \widehat{Q}_{A,i}^\vee, \widehat{d}_{A,j}]
    \]
    and
    \[
      [ \widehat{Q}_{A,i}, \widehat{d}_{A,j} ] = \delta_{ij}\, \widehat{d}_{A,j} 
      \,;\quad
      [ \widehat{Q}_{A,i}^\vee, \widehat{d}_{A,j}^\vee ] = \delta_{ij}\, \widehat{D}_{A,j}^\vee\, .
    \]
   \end{enumerate}
\end{lemma}

\begin{proof}
  We compute
  \begin{align*}
    [ \widehat{Q}_{A,i} , \widehat{d}_{A,k} ] 
    &= [\widehat{P}_{A,i}\, E_{A,i}\, E_{A,i}^\vee,
        \left(T_{A,j} + \psi_{A,j} \right) E_{A,j}] \\
    &= \delta_{ij}\, \widehat{P}_{A,i} \left( T_{A,j} + \psi_{A,j} \right) E_{A,j} \\
    &= \delta_{ij}\, \widehat{d}_{A,j},
  \end{align*}
  from which we see that
  \begin{equation*}
    [\widehat{Q}_{A,i}^\vee, \widehat{d}_{A,j}] = [E_{A,i}^\vee \, E_{A,i}, \widehat{d}_{A,j}] + \delta_{ij} \, \widehat{d}_{A,j} = \delta_{ij} \left( -\widehat{d}_{A,j} + \widehat{d}_{A,j} \right)=0\, .
  \end{equation*}
  Similarly,
  \begin{align*}
    [\widehat{Q}_{A,i}, \widehat{d}_{A,j}^\vee] 
    &= \widehat{P}_{A,i} \left( \theta_{A,j}^\vee + \psi_{A,j}^\vee \right) 
                        [E_{A,i}\, E_{A,i}^\vee, E_{A,j}^\vee ] \\
    &= -\delta_{ij} \, \widehat{P}_{A,i} 
       \left(\theta_{A,i}^\vee + \varphi_{A,i}^\vee \right) E_{A,i}^\vee\, .
  \end{align*}
  Since $\widehat{P}_{A,i}\, \theta_{A,i}^\vee (x^\gamma y^\lambda) \neq 0$ implies
  $(\gamma A^{-1})_i (\lambda A^{-T})_i > 0$, then
  \[
  (\gamma A^{-1})_i A_{ii} (\lambda A^{-T})_i > 0
  \] 
  and thus $\gamma A^{-1} \lambda^T > 0$. Therefore, $[ \widehat{Q}_{A,i}, \widehat{d}_{A,j}^\vee ] = 0$. As a consequence,
  \[
  [\widehat{Q}_{A,i}^\vee, \widehat{d}_{A,j}^\vee] 
  = [E_{A,i}^\vee\, E_{A,i}, \widehat{d}_{A,j}^\vee] = \delta_{ij} \widehat{d}_{A,j}^\vee,
  \]
  which concludes the proof of (2). For (1), we compute
  \begin{align*}
    \Delta^A\, Q_{A,i}^\vee &=  D^A\,P_{A,i}^\vee \otimes *^A\,e_i^\vee\,e_i  = \widehat{P}_{A^T,i}\, D^A \otimes E_{A^T,i}\, E_{A^T,i}^\vee \, *^A = \widehat{Q}_{A^T,i}\, \Delta^A\, ;\\
    \Delta^A\, Q_{A,i} &= \Delta \left( e_i\, e_i^\vee + Q_{A,i}^\vee \right) = \left( E_{A^T,i}^\vee\, E_{A^T,i} + \widehat{Q}_{A^T,i} \right) \Delta^A = \widehat{Q}_{A^T,i}^\vee \, \Delta^A\, .
  \end{align*}
\end{proof}

\begin{proposition}\label{prop:Fr_A-comm}
For each $A$ in $\BH_n(\F)$ the \emph{Frobenius endomorphism} defined by
\[
{\rm Fr}_A = \left((\Theta''_A\Theta_A')\otimes {\rm Id}_{\bigwedge (\mathbb K^n)}\right)p^{Q_A+Q_A^\vee}\,
\] 
is a chain map and
  \[
  \Delta^A \, \Fr_A = \Fr_{A^T} \, \Delta^A \, p^{2\,\ext - n}\, p^{-2 \widehat{Q}_A} \, p^{2 Q_A^\vee}.
  \]
\end{proposition}

\begin{proof} Since  $p^{Q_{pA}^\vee} \left( \Theta_A' \otimes \Id \right)
  = \left( \Theta_A' \otimes \Id \right) p^{Q_A^\vee}$, then ${\rm Fr}_A={\rm Fr}''_A{\rm Fr}'_A$ is a chain map. For the second statement,  $D^A \, \Theta_A'' \, \Theta_A' = \Theta_A'' \, \Theta_A' \, D^A$ implies
   \begin{align*}
    \Delta^A \, \Fr_A 
     &= \Delta^A \left( \Theta_A'' \, \Theta_A' \otimes \Id \right) p^{Q_A^\vee + Q_A} \\
     &= \Fr_{A^T}\, p^{-Q_{A^T} - Q_{A^T}^\vee} \, \Delta^A \, p^{Q_A^\vee + Q_A} \\
     &= \Fr_{A^T} \, \Delta^A \, p^{-\widehat{Q}_A - \widehat{Q}_A^\vee }\, p^{Q_A^\vee + Q_A}\\
     &= \Fr_{A^T} \, \Delta^A \, p^{2\,\ext - n} \, p^{-2\widehat{Q}_A} \, p^{2 Q_A^\vee}.
  \end{align*}
\end{proof}

\begin{theorem}\label{thm:TFr}
  Let $\#_A$ (respectively $\#_A^\vee$) be the operator on $\cS_A(\K)$
  diagonalized by monomials and such that the eigenvalue of $x^\gamma y^\lambda$ is
  the number of non-integer entries of $\gamma A^{-1}$ (respectively $\lambda A^{-T}$).
  If $\kappa$ is such that $(\kappa\pi)^{p-1}=p$, then the \emph{twisted Frobenius endomorphism}
  \[
  \TFr_A = \Fr_A \, (\kappa\pi)^{(p-1)(\#_A - \#_A^\vee)/2}
  \]
  is a chain map, and
  \[
  H(\Delta^A) H(\TFr_A) = H(\TFr_{A^T}) H(\Delta^A)\, .
  \]
\end{theorem}

\begin{proof}
  Since $(\kappa\pi)^{(p-1)(\#_A - \#_A^\vee)/2}$ is diagonalized by
  monomials and acts trivially on $\bigwedge(\mathbb K^n)$, it
  commutes with $d_A+d_A^\vee$. Therefore, $\TFr_A$ is a chain map. Using
  Proposition~\ref{prop:Fr_A-comm}, we calculate
  \begin{align*}
    \Delta^A\, \TFr_A
    &= \Delta^A\, \Fr_A\, (\kappa\pi)^{(p-1)(\#_A - \#_A^\vee)/2} \\
    &= \Fr_{A^T}\, \Delta^A\,  p^{2\,\ext - n} \, p^{-2\widehat{Q}_A} \, p^{2 Q_A^\vee}\, (\kappa\pi)^{(p-1)(\#_A - \#_A^\vee)/2} \\
    &= \TFr_{A^T}\, \Delta^A\, p^{2\,\ext - n} \, p^{-2\widehat{Q}_A} \, p^{2 Q_A^\vee}\, (\kappa\pi)p^{(p-1)(\#_A - \#_A^\vee)}, 
  \end{align*}
  where the last step follows from
  \[
  (\kappa\pi)^{-(p-1)(\#_{A^T} - \#_{A^T}^\vee)/2}\, \Delta^A = \Delta^A\, (\kappa\pi)^{(p-1)(\#_A - \#_A^\vee)/2}.
  \]  
  Therefore, the theorem is proven if the eigenvalues of
  \begin{equation}\label{eq:twisted-Fr-operators}
  2\,\ext - n -2\widehat{Q}_A + 2 Q_A^\vee \quad\text{and}\quad -(\#_A - \#_A^\vee)
  \end{equation}
  agree on a monomial basis $x^\gamma y^\lambda e^I$ for
  $H(\mathcal B_A^\lambda(\mathbb K))$ for each $\lambda\in G_A$. By
  Lemma \ref{lem:B_A-decomposition} and Corollary \ref{cor:BH_n-basis}
  one can choose generators of the form $x^{\gamma+I}y^\lambda e^I$,
  where $|I|=n-|J_\lambda^\vee|$ and $0\le (\lambda A^{-1})_i<1$ for
  all $i=1,\ldots,n$. In particular, the eigenvalue of
  $2\,\ext - n +2Q_A^\vee-\#_A^\vee$ on
  $x^{\gamma+I} y^\lambda e^I\in S(A)$ is $|I|$. On the other hand,
  inspection of the bases for the cohomology of chains and loops given
  in Corollary \ref{cor:BH_n-basis} shows that
  $(2{\widehat Q}_A-\#_A)=\ext$ on $S(A)$, which concludes the proof.
\end{proof}

\section{Examples}

\begin{example}
  Let $n=1$ and $A_{11}=2$. Then $W_A(x)=W_A^T(x)=x_1^2$ and
  $G_A=G_{A^T}=\mathbb Z/2\mathbb Z$. The exterior operators are
  $E_{A,1}=2\pi e_1$ and $E_{A,1}^\vee=\frac{1}{2\pi}
  e_1^\vee$.
  Moreover, $\mathcal R_A^0(\F)=\F[x_1]\oplus y_1^2\F[y_1^2]$ and
  $\mathcal R_A^1(\F)=y_1\F[y_1^2]$. The differentials are
  \begin{align*}
    d(x_1^{\gamma_1}) & =\gamma_1 x_1^{\gamma_1} e_1 +  2 \pi x_1^{\gamma_1+2} e_1;\\
    d^\vee(y_1^{\lambda_1}e_1) & = \frac{1}{2\pi}\lambda_1 y_1^{\lambda_1} + y_1^{\lambda_1+2}\, .
  \end{align*}
  It follows that $H(\mathcal B_A^0(\F))=\F x_1e_1$ and
  $H(\mathcal B_A^1(\F))=\F y_1$ are mapped one into the other by
  $\Delta^A$. The relations in cohomology are
  \begin{align*}
    x_1^{2k+1}e_1 &= (-2\pi)^{-1} (2k-1) x_1^{2(k-1)+1}e_1 =\ldots = (-2\pi)^{-k} (2k-1)!!\, x_1 e_1;\\
    y_1^{2k+1}& =  (-2\pi)^{-1} (2k-1) y_1^{2(k-1)+1} =\ldots = (-2\pi)^{-k} (2k-1)!!\, y_1\, .
  \end{align*}
  Let $(c_m)$ be the sequence of rational numbers defined by 
  \[
  e^{\pi\left( t^p - t \right)} = \sum_{m \geq 0}  c_m (-\pi)^m t^m.
  \]
  The action of the twisted Frobenius map in cohomology is thus
  \begin{align*}
    H({\rm TFr}_A)(x_1e_1) & = p (\kappa \pi)^{(p-1)/2} e^{\pi (x_1^{2p}-x_1^2)} x_1^p e_1 \\
    & = p(\kappa\pi)^{(p-1)/2} \sum_{m\geq0} c_m (-\pi)^m x_1^{2(m+\frac{p-1}{2})+1}e_1 \\
    & = p(\kappa\pi)^{(p-1)/2} \left(\sum_{m\geq0} c_m (-\pi)^{-\frac{p-1}{2}} 2^{-(m+\frac{p-1}{2})} (2(m-1)+p)!!\right) x_1e_1\\
    & = p \kappa^{(p-1)/2} \left( \left(\frac{p-1}{2}\right)! + \mathcal O(p)\right) x_1e_1\,. 
  \end{align*}
  Similarly, 
  \begin{align*}
    H({\rm TFr}_A)(y_1)& = p^2 (\kappa\pi)^{-(p-1)/2} \left(\sum c_m (-\pi)^{-\frac{p-1}{2}} 2^{-(m+\frac{p-1}{2})} (2(m-1)+p)!!\right) y_1 \\
    & =  p \kappa^{(p-1)/2}  \left( \left(\frac{p-1}{2}\right)! + \mathcal O(p)\right) y_1\, .
  \end{align*}
  Comparison with the non-commutative Weil conjectures of Kontsevich
  \cite{K} seems to suggest a further overall rescaling of $\TFr_A$. This is likely to be relevant for arithmetic applications. We hope to come back to this point in future work.
\end{example}

\begin{example} 
  Consider the dual chains $W_{A}(x)=x_{1}^{2}x_{2}+x_{2}^{3}$ and
  $W_{A^{T}}(x)=x_{1}^{2}+x_{1}x_{2}^{3}$. The elements of
  $G_{A}\cong\Z^{2}/\Z^{2}A^{T}$ and $G_{A^{T}}\cong\Z^{2}/\Z^{2}A$
  are given in Table~\ref{table:chain-2103-groups}. We can find basis
  elements $x^{\gamma}y^{\lambda}e^{I}$ of $\cC_{A}$ and $\cC_{A^{T}}$
  as described in the proof of Theorem~\ref{thm:TFr}.  Each row of
  Table~\ref{table:chain-2103-duality} contains a pair of elements
  dual under $\Delta^{A}$ (up to constants), as well as the
  eigenvalues of
  \[
  Q_A+Q_A^\vee
  \quad\text{and}\quad
  (\#_{A}-\#_{A}^{\vee})/2
  \]
  applied to $x^\gamma y^\lambda e^I$. Here we are using $*^{A}(e_{1}e_{2})=1$,
  $*^{A}(e_{2})=-E_{A^{T},1}=-2\pi e_{1}$ and
  \[
  *^{A}(1)=E_{A^{T},1}E_{A^{T},2}=(2\pi e_{1})(\pi e_{1}+3\pi e_{2})=6\pi^{2}e_{1}e_{2}.
  \]
  Note also that 
  \[
  \Delta^{A}(x_{1}^{2}x_{2}e_{1}e_{2})=y_{1}^{2}y_{2}\equiv3\pi x_{1}x_{2}^{3}e_{1}e_{2},
  \]
  since $(d_{A^{T}}+d_{A^{T}}^{\vee})(e_{1})=3\pi x_{1}x_{2}^{3}e_{1}e_{2}+y_{1}^{2}y_{2}$. 

  \renewcommand{\arraystretch}{1.5}
  \begin{table}
    \[
    \begin{array}[t]{|c|c|c|c|c|c|c|c|}
      \hline
      \multirow{2}{*}{$G_{A}$} & \lambda &  (0,0) &   (1,0) &  (1,1) &  (1,2) &  (2,1) &  (2,2)\\
      & \lambda A^{-T} & (0,0) & (\frac{1}{2},0) & (\frac{1}{3},\frac{1}{3}) & (\frac{1}{6},\frac{2}{3}) & (\frac{5}{6},\frac{1}{3}) & (\frac{2}{3},\frac{2}{3})\\
      \hline
      \multirow{2}{*}{$G_{A^{T}}$} &   \lambda &  (0,0) &  (0,1) &  (0,2) &  (1,1) &  (1,2) &  (1,3)\\
      & \lambda A^{-1} & (0,0) & (0,\frac{1}{3}) & (0,\frac{2}{3}) & (\frac{1}{2},\frac{1}{6}) & (\frac{1}{2},\frac{1}{2}) & (\frac{1}{2},\frac{5}{6})\\
      \hline
    \end{array}
    \]
    \caption{ Elements of $G_A$ and $G_{A^T}$ for $W_A(x) = x_1^2 x_2 + x_2^3$. }
    \label{table:chain-2103-groups}
  \end{table}
  \renewcommand{\arraystretch}{1.0}

  \renewcommand{\arraystretch}{1.5}
  \begin{table}
    \[
    \begin{array}[t]{|c|c|c|c|c|c|}
      \hline
      \cC_{A} & Q_A + Q_A^\vee & (\#_{A}-\#_{A}^{\vee})/2 & \cC_{A^{T}} & Q_{A^T} + Q_{A^{T}}^{\vee} & (\#_{A^{T}}-\#_{A^{T}}^{\vee})/2\\
      \hline
      
      x_{1}x_{2}e_{1}e_{2} & 2 & 1 & y_{1}y_{2} & 4 & -1\\
      x_{1}x_{2}^{2}e_{1}e_{2} & 2 & 1 & y_{1}y_{2}^{2} & 4 & -1\\
      
      x_{1}x_{2}^{3}e_{1}e_{2} & 2 & 1 & y_{1}y_{2}^{3} & 4 & -1\\
      x_{1}^{2}x_{2}e_{1}e_{2} & 2 & 0 & x_{1}x_{2}^{3}e_{1}e_{2} & 2 & 0\\
      
      x_{2}y_{1}e_{2} & 3 & 0 & x_{1}y_{2}e_{1} & 3 & 0\\
      x_{2}^{2}y_{1}e_{2} & 3 & 0 & x_{1}y_{2}^{2}e_{1} & 3 & 0\\
      
      y_{1}y_{2} & 4 & -1 & x_{1}x_{2}e_{1}e_{2} & 2 & 1\\
      y_{1}y_{2}^{2} & 4 & -1 & x_{1}x_{2}^{2}e_{1}e_{2} & 2 & 1\\
      
      y_{1}^{2}y_{2} & 4 & -1 & x_{1}^{2}x_{2}e_{1}e_{2} & 2 & 1\\
      y_{1}^{2}y_{2}^{2} & 4 & -1 & x_{1}^{2}x_{2}^{2}e_{1}e_{2} & 2 & 1\\
      \hline
    \end{array}
    \]

    \protect\caption{Duality between $\cC_A$ and $\cC_{A^T}$ for $W_A(x) = x_1^2 x_2 + x_2^3$. }
    \label{table:chain-2103-duality}
  \end{table}
  \renewcommand{\arraystretch}{1.0}
  
  We now turn to writing $\TFr_{A}(x^{\gamma}y^{\lambda}e^{I})$ in
  terms of this basis for a few elements. Since for any $x^{\gamma}$,
  \begin{align*}
    (\theta_{A,1}+\varphi_{A,1})(x^{\gamma+e_{1}A}) & =\gamma_{1}x^{\gamma}+\pi\left(2x^{\gamma+e_{1}A}\right);\\
    (\theta_{A,2}+\varphi_{A,2})(x^{\gamma+e_{2}A}) & =\gamma_{2}x^{\gamma}+\pi\left(x^{\gamma+e_{1}A}+3x^{\gamma+e_{2}A}\right),
  \end{align*}
  in $H\left(\cB_{A}^{\lambda}(\F)\right)$ we have the relation
  \[
  \gamma x^{\gamma} y^\lambda e^{I}=(-\pi)(x^{\gamma+e_{1}A} y^\lambda e^I,x^{\gamma+e_{2}A} y^\lambda e^I)A,
  \]
  which implies for $i = 1 , 2$ that 
  \[
  x^{\gamma+e_{i}A}y^\lambda e^{I}=(-\pi)^{-1}\left(\gamma A^{-1}\right)_{i}x^{\gamma} y^\lambda e^{I}.
  \]
  Therefore, for $i = 1,2$,
  \begin{align}
    x^{\gamma+k_{i}e_{i}A} y^\lambda e^{I}= & (-\pi)^{-1}\left((\gamma+(k_{i}-1)e_{i}A)A^{-1}\right)_{i}x^{\gamma+(k_{i}-1)e_{i}A} y^\lambda e^{I}\nonumber \\
    = & (-\pi)^{-2}\left((\gamma A^{-1})_{i}+(k_{i}-1)\right)\left((\gamma A^{-1})_{i}+(k_{i}-2)\right)x^{\gamma+(k_{i}-2)e_{i}A} y^\lambda e^{I}\nonumber \\
    = & (-\pi)^{-k_{i}}\left((\gamma A^{-1})_{i}\right)_{(k_{i})}x^{\gamma} y^\lambda e^{I}.\label{eq:reduction-x}
  \end{align}
  Take $x_{1}x_{2}\,e_{1}e_{2}$ so that
  $\gamma=(1,1)$ and $\gamma A^{-1}=(\frac{1}{2},\frac{1}{6})$.
  Suppose that $p$ is a prime such that $6\mid(p-1)$. Then we can write 
  \[
  (p,p)=(1,1)+\left(\frac{p-1}{2}, \frac{p-1}{6} \right) A,
  \]
  which using equation (\ref{eq:reduction-x}) gives
  \begin{align*}
    \TFr_A(x_1 x_2\,e_1 e_2) 
    &= p^{2} (\kappa \pi)^{p-1} x_1^p x_2^p Z_A(x) e_1 e_2 \\
    & =p^3 x_{1}^{p}x_{2}^{p}\left(\sum_{k_{1}\geq0}(-\pi)^{k_{1}}c_{k_{1}}x^{k_{1}e_{1}A}\right)\left(\sum_{k_{2}\geq0}(-\pi)^{k_{2}}c_{k_{2}}x^{k_{2}e_{2}A}\right)e_{1}e_{2}\\
    & =p^3 (-\pi)^{-\frac{2(p-1)}{3}} \left(\sum_{k_{1},k_{2}\geq0}c_{k_{1}}c_{k_{2}}
       \left(\frac{1}{2}\right)_{(k_{1} + \frac{p-1}{2})}\left(\frac{1}{6}\right)_{(k_{2} + \frac{p-1}{6})}\right) x_{1}^{1}x_{2}^{1}\,e_{1}e_{2},
  \end{align*}
  where we have used the fact that
  $Z_{A^{T}}(y)=1+\mathcal{O}(y_{1},y_{2})$. Next, consider
  \[
  \TFr_A( x_2^2 y_1 e_2 ) = p^3  e^{\pi \left( x_2^{3p} - x_2^3 \right)}  e^{\pi  \left( y_1^{2p} - y_1^p \right)}  x_2^{2p} y_1^p \, e_2.
  \]
  By equation~(\ref{eq:cB_A^lambda-decomposition}), in cohomology we have the relation
  \[
  y^{\lambda + k_1' e_1 A^T} = (-\pi)^{-k_1'} \left( (\lambda A^{-T} )_1 \right)_{(k_1')} y^\lambda 
    = (-\pi)^{k_1'} \left( \frac{3\lambda_1 - \lambda_2}{6} \right)_{(k_1')} y^\lambda,
  \]
  which if $6 \mid (p-1)$ implies that
  \begin{equation*}
    \TFr_A ( x_2^2 y_2 e_2 ) 
    =p^3 (-\pi)^{-\frac{7(p-1)}{6}} \left( \sum_{k_1',k_2 \geq 0}
      c_{k_1'} c_{k_2}
      \left(\frac{1}{2}\right)_{\left( k_1' + \frac{p-1}{2} \right)}
      \left( \frac{2}{3}
      \right)_{\left( k_2 + \frac{2(p-1)}{3} \right)} \right) x_2^2
    y_1 e_2.
  \end{equation*}
  \end{example}


\begin{bibdiv}
\begin{biblist}
\bib{B1}{article}{
   author={Borisov, Lev A.},
   title={Vertex algebras and mirror symmetry},
   journal={Comm. Math. Phys.},
   volume={215},
   date={2001},
   number={3},
   pages={517--557},
}

 

  \bib{B2}{article}{
    title={{B}erglund-{H}{\"u}bsch mirror symmetry via vertex algebras},
    author={Borisov, Lev A.},
    journal={Communications in Mathematical Physics},
    volume={320},
    number={1},
    pages={73--99},
    year={2013},
    publisher={Springer}
  }

  \bib{BH}{article}{
    title={A generalized construction of mirror manifolds},
    author={{Berglund}, Per}
    author={{H{\"u}bsch}, Tristan},
    journal={Nuclear Physics B},
    volume={393},
    number={1},
    pages={377--391},
    year={1993},
    publisher={Elsevier}
  }



  \bib{GP}{article}{
    title={Duality in {C}alabi--{Y}au moduli space},
    author={{Greene}, Brian R.},
    author={{Plesser}, Ronen M.},
    journal={Nuclear Physics B},
    volume={338},
    number={1},
    pages={15--37},
    year={1990},
    publisher={Elsevier}
  }


    
     \bib{K}{article}{
   author={Kontsevich, Maxim},
   title={XI Solomon Lefschetz Memorial Lecture series: Hodge structures in
   non-commutative geometry},
   conference={
      title={Non-commutative geometry in mathematics and physics},
   },
   book={
      series={Contemp. Math.},
      volume={462},
      publisher={Amer. Math. Soc., Providence, RI},
   },
   date={2008},
   pages={1--21},
}


    \bib{Kre}{article}{
      title={The mirror map for invertible {LG} models},
      author={Kreuzer, Maximilian},
      journal={Physics Letters B},
      volume={328},
      number={3},
      pages={312--318},
      year={1994},
      publisher={Elsevier}
    }

    \bib{KS}{article}{
      title={On the classification of quasihomogeneous functions},
      author={Kreuzer, Maximillian},
      author={Skarke, Harald},
      journal={Communications in mathematical physics},
      volume={150},
      number={1},
      pages={137--147},
      year={1992},
      publisher={Springer}
    }

    \bib{Kra}{thesis}{
   author={Krawitz, Marc},
   title={FJRW rings and Landau-Ginzburg mirror symmetry},
   note={Thesis (Ph.D.)--University of Michigan},
   date={2010},
}

\bib{M}{book}{
   author={Monsky, Paul},
   title={$p$-adic analysis and zeta functions},
   series={Lectures in Mathematics, Department of Mathematics, Kyoto
   University},
   volume={4},
   publisher={Kinokuniya Book-Store Co., Ltd., Tokyo},
   date={1970},
   pages={iv+117},
}


   
    
  \bib{P}{thesis}{
   author={Perunicic, Andrija},
   title={Arithmetic Aspects of Berglund-Hubsch Duality},
   note={Thesis (Ph.D.)--Brandeis University},
   date={2013},
}


    \bib{SS}{article}{
      title={Twisted {de} {R}ham cohomology, homological definition of the integral and ``physics over a ring''},
      author={Schwarz, Albert},
      author={Shapiro, Ilya},
      journal={Nuclear Physics B},
      volume={809},
      number={3},
      pages={547--560},
      year={2009},
      publisher={Elsevier}
    }
    
    \bib{W}{article}{
   author={Wan, Daqing},
   title={Mirror symmetry for zeta functions},
   note={With an appendix by C. Douglas Haessig},
   conference={
      title={Mirror symmetry. V},
   },
   book={
      series={AMS/IP Stud. Adv. Math.},
      volume={38},
      publisher={Amer. Math. Soc., Providence, RI},
   },
   date={2006},
   pages={159--184},
}

\end{biblist}
\end{bibdiv}
\end{document}